\documentclass[final,1p,times]{elsarticle}
\usepackage{color}
\definecolor{grey}{gray}{0.75}

\newtheorem{theorem}{Theorem}
\newtheorem{lemma}[theorem]{Lemma}
\newdefinition{problem}[theorem]{Problem}
\newdefinition{definition}[theorem]{Definition}
\newproof{proof}{Proof}
\newproof{firstproof}{First proof}
\newproof{secondproof}{Second proof}

\newcommand\old[1]{}

\def\qed{{\hfill$\Box$}}

\def\P{{\mathcal{P}}}
\def\R{\mathbf{root}}
\def\cost{\mathrm{cost}}

\begin{document}

\begin{frontmatter}

\title{Balanced Vertices in Trees and a Simpler Algorithm to Compute the Genomic Distance\tnoteref{tn1}}

\author[renyi]{P{\'e}ter L. Erd{\H o}s\fnref{fn1}}
\ead{elp@renyi.hu}

\author[renyi]{Lajos Soukup\fnref{fn2}}
\ead{soukup@renyi.hu}

\author[bi]{Jens Stoye\corref{cor}}
\ead{stoye@techfak.uni-bielefeld.de}

\address[renyi]{Alfr\'ed R{\'e}nyi Institute of Mathematics,
  Budapest, P.O.~Box 127, H-1364 Hungary}
\address[bi]{Universit\"at Bielefeld, Technische Fakult\"at,
  AG Genominformatik, 33594 Bielefeld, Germany}

\fntext[fn1]{Partly supported by Hungarian NSF, under contract Nos. AT048826 and K68262.}
\fntext[fn2]{Partly supported by Hungarian NSF, under contract Nos. K61600 and K68262.}
\tnotetext[tn]{This research was carried out when the first and second authors visited the third author at Bielefeld University with the support of the Hungarian Bioinformatics MTKD-CT-2006-042794, Marie Curie Host Fellowships for Transfer of Knowledge.}
\cortext[cor]{Corresponding author}

\begin{abstract}
This paper provides a short and transparent solution for the covering cost of white--grey trees which play a crucial role in the algorithm of Bergeron {\it et al.}\ to compute the rearrangement distance between two multichromosomal genomes in linear time ({\it Theor. Comput. Sci.}, 410:5300--5316, 2009). In the process it introduces a new {\em center}
notion for trees, which seems to be interesting on its own.
\end{abstract}


\end{frontmatter}

\section{Introduction}\label{sec:intro}
\noindent Computational comparative genomics is a subdiscipline of computational biology in which the relationships between two or more genomes are studied by computational means. A highly relevant question in this field is the calculation of the minimum number of rearrangement operations (reversals, translocations, fusions and fissions) that are necessary to transform one given genome into another one, the so-called \emph{genome rearrangement problem} \cite{hp95}.

The \emph{white--grey tree cover problem} studied in this paper (formally defined in Section~\ref{sec:prob}) arises as a subproblem of the genome rearrangement problem, and so far only an unsatisfactory (and not self-contained) solution exists~\cite{bms09}. The main goal of this paper is to give a short solution of the problem and to correct some omissions and discrepancies of the original formulation. (In Section \ref{sec:proof} we point out some cases where the original formulation fails.) Moreover, it gives rise to a combinatorial problem on trees, detailed in Section~\ref{sec:bal}, that seems to be interesting on its own. Since one of our main concerns here is brevity, we (usually) don't give detailed proofs of easy facts, which are not essential for our main goal.

\section{Problem definition}\label{sec:prob}
\noindent A \emph{white--grey tree} is a rooted tree with (white or grey) colored and uncolored vertices. The $\R$ is uncolored, some children of the $\R$ are grey (some of them can be leaves), all leaves which are not children of the $\R$ are white. All uncolored vertices (with the possible exception of the root) are branching points.

A system of paths in a white--grey tree is a {\em colored cover} if:
\begin{enumerate}[{\rm (i)}]
\item Each path has colored endpoints. One vertex alone may constitute a path.
\item Each colored vertex is covered with path(s).
\end{enumerate}

The {\em cost} of a path $P$ is denoted by $\cost(P)$ and is defined as follows:
\begin{enumerate}[{\rm (i)}]
\item $P$ is {\em short} if it has exactly one vertex. Then $\cost(P)=1.$
\item $P$ is {\em grey} if its endpoints are grey vertices (then the third vertex is the uncolored $\R$). Then $\cost(P)=1.$
\item $P$ is {\em long} otherwise. Its cost is $\cost(P)=2.$
\end{enumerate}

\begin{definition}
The {\em cost of a path system} $\P$ is the sum of the individual costs: $\cost(\P):= \sum_{P\in \P} \cost(P)$. A colored cover $\P$ is an {\em optimal} one for a given white--grey tree $T$ if it has minimal cost among all possible colored covers, denoted by $\cost(T)$.
\end{definition}

\begin{problem}[White--grey tree cover problem]
Given a white--grey tree $T$, compute $\cost(T)$.
\end{problem}
The main result of this paper is a simple way to calculate the exact cost of an optimal cover. We are not quite ready to formalize the main result (without some further observations it would require a detailed case analysis), but we mention here a well known fact~\cite{hp95}:

\begin{lemma}\label{lm:bounds}
Let $T$ be a white--grey tree with $w$ white and $g$ grey leaves, then:
\begin{displaymath}
\hspace{3cm} w + \lceil g/2 \rceil \;\le\; \cost(T) \;\le\; w + \lceil g/2 \rceil + 1. \hspace{4cm}\Box
\end{displaymath}

\end{lemma}

\section{Balanced vertices in trees}\label{sec:bal}
\noindent In this section we prove a useful tool for (unrooted) trees which seems to be interesting in its own. In tree $T'$ denote by $P_{u,v}$ the unique path between vertices $u$ and $v$.

\begin{theorem}\label{th:balanced}
Let $T'$ be a tree with $2n$ leaves. Then there exists a vertex $v\in V(T')$ and a bijection among the leaves $\alpha: \mathcal{L} \rightarrow \mathcal{L}$ such that the path system $P_{\ell, \alpha(\ell)}$ (where $\ell \in \mathcal{L}$) covers each vertex in $T'$ and all paths contain $v.$
\end{theorem}
We offer here two proofs. One gives a very simple algorithm to construct such a cover, but it clearly cannot provide all possible solutions. The second proof is based on a necessary and sufficient reformulation of the statement.

\begin{firstproof}
Consider an embedding of our tree into the plane and enumerate the leaves in a counter clockwise fashion. One way to obtain such a numbering is to fix a leaf as a root and take the left-to-right, depth first traversal of the tree which conforms with the embedding.

Now we define our bijection with the formula $\alpha: \ell_i \mapsto \ell_{i+n} \textrm{ mod } 2n$. Considering any two such paths, their endpoints alternate along the circle which contains the leaves in increasing order. Therefore these two paths clearly intersect each other.

As it is well known (its very first proof is due to Gy\'arf\'as and Lehel, \cite{gl70}) if (in a tree) a set of paths does not contain two disjoint paths, then all the paths share a common vertex $v$. And because these paths connect $v$ to the leaves, they cover all edges of the tree.
\qed
\end{firstproof}

If $T'$ is a fully balanced tree, then no matter what is the embedding in the previous proof, two close leaves will be paired with two close leaves. Therefore there are clearly solutions which cannot be obtained with the previous method. In the remaining part of this section we sketch a proof which is able to find all possible solutions:

\begin{secondproof}
For each vertex-edge pair $(v,e)$ denote by $\delta (v,e)$ the number of leaves $\ell$ in $T'$ such that $P_{v,\ell}$ contains the edge $e$ (where $v\in V(T')$, $e\in E(T')$ and $v \in e$). Furthermore, denote by $E(v)$ the set of edges that contain $v$.

In the configuration required by Theorem~\ref{th:balanced}, vertex $v$ clearly satisfies the inequality:
\begin{equation}\label{eq:bal}
\forall e\in E(v) \quad : \quad \delta(v,e) \le \sum \big\{ \delta(v,f) : f \in E(v), f \ne e \big\}.
\end{equation}
Such a vertex $v\in T'$ is called a {\em balanced} vertex. (If a vertex-edge pair does not satisfy this inequality, then the pair is called {\em oversaturated}.) As a matter of fact, this property is just equivalent to the existence of the required cover:

\begin{lemma}\label{lm:path}
Let $T'$ be a tree with $2n$ leaves, $n\geq 1$. Then for any balanced vertex $v$ there exists a bijection $\alpha$ such that the paths $P_{\ell,\alpha(\ell)}$ cover all edges, and all paths contain $v.$
\end{lemma}

\noindent The easy proof is left to the diligent reader. (One can argue, for example, with mathematical induction.) A balanced vertex in a tree is similar to the well-known notion {\em center} of the tree, but while a center is usually (almost) unique, there may be several balanced vertices.

The following observation completes our second proof of Theorem~\ref{th:balanced}.
\begin{lemma}\label{lm:bal}
Any tree $T'$ with an even number of leaves contains a balanced vertex.
\end{lemma}

\begin{proof} (Sketch) Assume that a particular vertex $v$ is not balanced. Then there exists an edge $e\in E(v)$ such that the pair $(v,e)$ is oversaturated. We repeat the process with the other end of that edge. If this vertex is not balanced again, then it will provide another oversaturated pair. The finiteness of the graph finishes the proof of the Lemma and this also completes the second proof of Theorem~\ref{th:balanced}. \qed
\end{proof}
\end{secondproof}

The flexibility in the pairing algorithm clearly can provide any possible bijection $\alpha$. It is also interesting to recognize that one can find a suitable balanced vertex quickly:
\begin{lemma}\label{lm:fast}
Let $T'$ be a tree with $2n$ leaves. Then there is a linear (in the number of leaves) time algorithm to find a balanced vertex in $T'$.
\end{lemma}
\noindent This proof is left to the reader again. A simple dynamic programming algorithm suffices. \qed

\section{Optimal colored covers}\label{sec:proof}
\noindent We are ready to determine the cost of an optimal cover for the white--grey tree $T.$ We say a path in the cover is a {\em mixed} path if it contains at least two colored vertices, exactly one of the colored vertices is a grey leaf. We will use the notation $T_w$ for the subtree derived from $T$ by deleting all grey leaves and their edges, and the $\R$ if it would become a leaf. Furthermore for a path $P$ in $T$ we will use the notation $P\restriction T_w$ to denote the {\em trace} of $P$ on $T_w$, i.e.\ the restriction of $P$ to the nodes of $T_w$ with the extra condition that in the truncated path we delete the starting (if any) uncolored vertices. We extend this notation to the trace of a path system $\P$, $\P\restriction T_w$.

Our general strategy to determine an optimal colored cover is to build it from an optimal colored cover of the subtree $T_w.$ To do that we are going to exploit certain properties -- described in the following result -- of optimal solutions having a minimum number of mixed paths.
\begin{theorem}\label{tm:nice}
Every white--grey tree $T$ has an optimal colored cover $\P$ such that
\begin{enumerate}[(1)]
\item $\P$ contains at most $2$ mixed paths,
\item $\P\restriction T_w$ is an optimal cover of $T_w$,
\item for each mixed path $P\in\P$, $\cost(P)=\cost(P\restriction T_w)$ and so $P\restriction T_w$ is either a long path, or a short path consisting of a single grey leaf.
\end{enumerate}
\end{theorem}
\begin{proof}
(1) Let $\P$ be an optimal cover with a minimum number of mixed paths. Assume on the contrary that $\P$ contains three mixed paths: $P_1, P_2$ and $P_3$, where $P_i$ is a path from the grey leaf $g_i$ to the colored vertex $c_i\in T_w$. (If two paths cover the same grey leaf then deleting that leaf from one of the paths decreases the number of the mixed paths in the cover. So we may assume that the grey leaves are pairwise distinct.)

Let the path $P$ be the intersection of the paths $P_1,P_2, P_3$. Clearly $P$ is a path from the {\bf root} to some $c\in T_w$. (It is clear that $c$ may be the $\R$ itself). Since $c$ is the ``last'' point of the intersection, we can assume that the unique sub-paths $P_{c_1, c}$ and $P_{c_2, c}$ are edge disjoint (and of course vertex-disjoint except vertex $c$). Then replace the paths $\{P_1, P_2, P_3\}$ in $\P$ with the paths $\{ P_{g_1,g_2}, P_{c_1,c_2}, P_3 \}$ to obtain a path cover $\P'$. But $\cost(\P')\le \cost(\P)$ and $\P'$ contains less mixed paths than $\P$ -- a contradiction.

(2) So we have an optimal cover which contains at most two mixed paths. If its trace is not optimal then consider the following cover $\mathcal{Q}$: cover $T_w$ optimally (this has cost at least 1 smaller than the trace of the original cover had), keep the paths from $\P$ which do not intersect $T_w$ and finally cover the (at most two) grey vertices that were covered by the mixed path(s) with a path whose cost is 1. Then the cost of $\mathcal{Q}$ is less than or equal to the cost of $\P$, and $\mathcal{Q}$ does not contain any mixed path.

(3) Assume that $2 = \cost(P) > \cost(P\restriction T_w) = 1$ for a mixed path $P\in \P$. The restriction $P \restriction T_w$ should be a short white path covering vertex $u$, and $P$ is a path $P_{u,g}=(u,\R,g)$ for a grey leaf $g$. Replacing the path $P$ with two short paths covering $u$ and $g$ resp.\ keeps the cost of the cover, but decreases the number of mixed paths. \qed
\end{proof}
A cover $\P$ is {\em nice} iff it satisfies the requirements of Theorem \ref{tm:nice}. Let $P$ be a path in $T_w$. We say that path $P$ is {\em free} iff $P$ can be extended to path $P'$ such that $P'$ contains a grey leaf while $\cost(P) = \cost(P')$ holds. Theorem \ref{tm:nice} implies the following statement:

\begin{lemma}\label{lm:reduction}
Assume that $T$ is a white--grey tree which has $g$ grey leaves.
\begin{displaymath}
\cost(T)= \cost(T_w) + \max\left\{0, \left\lceil \frac{g-f}2 \right\rceil \right\},
\end{displaymath}
where $f$ is the maximal number of free paths in a nice optimal cover of $T_w$. \qed
\end{lemma}

\medskip\noindent Next we should solve the white--grey tree cover problem for the subtree $T_w.$ Therefore we first solve the problem for trees where (essentially) all leaves are white. In what comes, we will say a leaf is {\em short} if it is adjacent to a branching vertex.
\begin{lemma}\label{lm:white-cost}
Let $T'$ be a white--grey tree with $w$ colored leaves but without a grey vertex or with exactly one grey leaf. Then the minimal cost of a colored cover is:
\begin{equation}\label{eq:white-cost}
\cost(T')=
\left\{
  \begin{array}{ll}
    w+1, & \hbox{if } w \hbox{ is odd and there is no short leaf} ;\\
    w, & \hbox{otherwise.}
  \end{array}
\right.
\end{equation}
\end{lemma}
\begin{proof}
Since we have at most one grey leaf, we can not use a ``cheap'' grey-grey path to cover it. So we can change the color of that vertex into white without changing the cost of the tree and thus assume that all leaves are white.

If the number of leaves is even, then the result is a direct consequence of Lemma~\ref{lm:bounds} and Theorem \ref{th:balanced}. If the number of leaves is odd, but there is a short leaf, then we cover that leaf with a short path. Deleting it from the tree we are back to the previous case.

Finally assume that $w$ is odd but $\cost(T')=w.$ Then each leaf is covered once in an optimal cover, and one of them is covered by a short path. If this leaf is not a short one, however, then its colored neighbor is not covered, a contradiction.
For simplicity we fix: in this case the constructed optimal cover contains a long path which does not cover any branching vertex. This path will be called a {\em half--path}. \qed
\end{proof}
\noindent Let's remark that Lemma~\ref{lm:white-cost} for white--only trees is certainly not new: actually it was proved as early as 1995 (see \cite{hp95}). But the consideration of more general white--grey trees raises several problematic issues. One of them is that in the literature, known to the authors, grey vertices which are not leaves have not been studied. However, the white--grey trees are constructed in connection with the genome rearrangement problem (\cite{bms09}) and grey vertices can appear in non-leaf positions.

Another problem that paper \cite{bms09} fails to determine is the exact cost of a minimal colored cover for some cases. Here we give only one of them. (The references relate to the relevant sections of that paper.) Assume that the $\R$ of $T$ has two neighbors: one is a grey leaf ($g=1$), and the other one is a branching vertex. Furthermore assume that $w$ is odd, and no white leaf is short. Then we are in the scope of Theorem 5 of \cite{bms09}. Since $g$ is odd and $T_c$ is a {\em fortress} or {\em junior fortress}, we are to apply the case ``otherwise'' of Theorem 5. That formula now gives: $\cost(T)= w + \lceil g/2 \rceil +1 = w+2$ while the proper cost is only $w + \lceil g/2 \rceil = w+1$.

\bigskip\noindent Before we give our main result we introduce one more notion: when among the children of the $\R$ there is exactly one child that is not a grey leaf, then the (colored) vertices between the $\R$ and the first branching point are called {\em dangerous}.
\begin{theorem}\label{tm:counting}
Let $T$ be a white--grey tree with $g$ grey and $w$ white leaves. Let $T_w$ be derived from $T$ by deleting the grey leaves (and the $\R$ if it would become a leaf). \\
{\rm (1)} If $T$ does not have any dangerous vertex then
\begin{displaymath}
  \cost(T)=\left\{
    \begin{array}{ll}
      w + 1 + \left\lceil\frac{g-1}2 \right\rceil, &\mbox{if $w$ is odd and there is no short leaf in $T_w$;}\\
      w + \left\lceil\frac{g}2 \right\rceil, &\mbox{otherwise}.
    \end{array}
  \right.
\end{displaymath}\\
{\rm (2)} If $T$ has some dangerous vertices (and $T_w$ has $(w+1)$ leaves) then
\begin{displaymath}
  \cost(T)=\left\{
    \begin{array}{ll}
      (w+1) + 1 + \max\left\{ 0, \left\lceil\frac{g-2}2 \right\rceil\right\}, &\mbox{if $(w+1)$ is odd and there is no short leaf in $T_w$;}\\[2pt]
      (w+1) + \left\lceil \frac{g}2 \right\rceil, &\mbox{if } \left \{\parbox{7cm}{$(w+1)$ is odd, there is only one short leaf in $T_w$, and that leaf is white and dangerous;}\right. \\[2pt]
      (w+1) + \left\lceil\frac{g-1}2 \right\rceil, &\mbox{otherwise}.
    \end{array}
  \right.
\end{displaymath}
\end{theorem}
\begin{proof}
(1) Assume that $w$ is odd and there is no short leaf in $T_w$. Then, due to Lemma~\ref{lm:white-cost}, $\cost(T_w)=w+1$ and the half--path of the solution is clearly free, so $f=1$. Otherwise $\cost(T_w)=w$, but we have no free path at all, so $f=0$. In both cases apply Lemma~\ref{lm:reduction}.

\medskip \noindent(2) {\bf Case 1}: Assume that $(w+1)$ is odd and there is no short leaf in $T_w$. Then $\cost(T_w)=(w+1)+1$. In the derived optimal cover $\P_w$ of $T_w$ a leaf-leaf type long path $P$ covers the dangerous vertices. Then the path $P$ and the half--path of $\P_w$ are free, so $f=2$. Then apply Lemma \ref{lm:reduction}

\noindent {\bf Case 2}: Assume now that $(w+1)$ is odd, there is only one short leaf $\ell$ in $T_w$, and that leaf is white and dangerous. Then $\cost(T_w)=w+1$. Moreover, an optimal cover of $\P_w$ should contain the short (non--free) path covering $\ell$, and there is no other free path, thus $f=0$. Lemma~\ref{lm:reduction} finishes the case.

\noindent {\bf Case 3}: The ``otherwise'' cases: Assume first that $(w+1)$ is odd, there is only one short leaf $\ell$, and that vertex is a grey (therefore also a dangerous) vertex. Then $\cost(T_w)=w+1$. Moreover, an optimal cover of $T_w$ should contain the short path covering $\ell$. But $\ell$ is grey, so the path covering $\ell$ is free. Thus $f=1$. Then apply Lemma~\ref{lm:reduction}.

Assume now that $(w+1)$ is odd, and there is a short leaf $\ell$ which is not dangerous. Then there is an optimal cover of $T_w$ in which the dangerous vertices are covered by a long path $P$. Then $P$ is free, so $f\ge 1$. Since $\cost(T_w)= (w+1)$, in an optimal cover $\P_w$ of $T_w$ there is only one long path which is free because all long paths in $\P_w$ contain two leaves. So $f\le 1$, i.e. $f=1$. Now apply Lemma \ref{lm:reduction}.

Assume finally that $(w+1)$ is even. Then $\cost(T_w)=(w+1)$, in an optimal cover of $T_w$ there is only one long path which is free because all long paths in an optimal cover should contain two leaves. So $f\le 1$. However, there is an optimal cover of $T_w$ containing a long path which covers the dangerous vertices. So $f\ge 1$. Thus $f=1$. Now apply Lemma \ref{lm:reduction}. \qed
\end{proof}

\bibliographystyle{elsarticle-num}
\bibliography{biblio}

\begin{thebibliography}{1}
\expandafter\ifx\csname url\endcsname\relax
  \def\url#1{\texttt{#1}}\fi
\expandafter\ifx\csname urlprefix\endcsname\relax\def\urlprefix{URL }\fi
\expandafter\ifx\csname href\endcsname\relax
  \def\href#1#2{#2} \def\path#1{#1}\fi

\bibitem{hp95}
S.~Hannenhalli, P.~A. Pevzner, Transforming men into mice (polynomial algorithm
  for genomic distance problem), in: Proc. 36th Annu. Symp. Found. Comput.
  Sci., FOCS 1995, IEEE Press, 1995, pp. 581--592.

\bibitem{bms09}
A.~Bergeron, J.~Mixtacki, J.~Stoye, A new linear time algorithm to compute the
  genomic distance via the double cut and join distance, Theor. Comput. Sci.
  410 (2009) 5300--5316.

\bibitem{gl70}
A.~Gy\'arf\'as, J.~Lehel, A helly-type problem in trees, in: P.~Erd\H{o}s,
  A.~R\'enyi, V.~T. S\'os (Eds.), Combinatorial Theory and its Applications II,
  Vol.~4 of Colloquia Mathematica Scocietatis J\'anos B\'olyai, North Holland,
  1970, pp. 571--584, (Balatonf\"ured, 1969).

\end{thebibliography}

\end{document}